\documentclass[11pt]{article}
\usepackage{amssymb,amsmath}
\usepackage{epsfig}
\usepackage{palatino}
\usepackage{graphicx}
\usepackage{textcomp}
\usepackage{hyperref}
\usepackage{amsthm}

\newtheorem{theorem}{Theorem}
\newtheorem*{theorem*}{Theorem}
\newtheorem*{lemma*}{Lemma}
\newtheorem*{algorithm*}{Algorithm}
\newtheorem{proposition}{Proposition}

\newtheorem{definition}{Definition}
\newtheorem{claim}{Claim}
\newtheorem{lemma}{Lemma}

\newtheorem{corollary}{Corollary}
\newtheorem{fact}{Fact}

\def\OO{{\rm O}}

\def\Z{{\mathbb{Z}}}
\def\R{\mathbb{R}}

\def\F{\mathbb{F}}
\def\mod{\mbox{mod}}

\def\poly{{\rm poly}}
\def\log{{\rm log}}
\def\gapsvp{{\rm gapSVP}}

\def\svp{{\rm SVP}}
\def\cvp{{\rm CVP}}
\def\gdd{{\rm GDD}}
\def\sivp{{\rm SIVP}}
\def\sis{{\rm SIS}}

\def\modn{(\mod \ N)}

\newcommand{\be}{\begin{eqnarray}}
\newcommand{\ee}{\end{eqnarray}}

\newcommand\round[1]{{\lfloor #1 \rceil}}

\def\P{{\sf{P}}}

\def\dist{{\rm dist}}

\def\rsis{\rm rSIS}

\newcommand{\eps}{\varepsilon}
\renewcommand{\epsilon}{\varepsilon}

\begin{document}

\title{The Systematic Normal Form of Lattices}

\author{Lior Eldar\thanks{Center for Theoretical physics, MIT} \and
Peter W. Shor\thanks{Department of Mathematics and Center for Theoretical physics, MIT}
}

\maketitle

\abstract{
We introduce a new canonical form of lattices called the {\it systematic normal form} (SNF).
We show that for every lattice there is an efficiently computable "nearby" SNF lattice, 
such that for any lattice one can solve lattice problems on its "nearby" SNF lattice, and translate
the solutions back efficiently to the original lattice.
The SNF provides direct connections between arbitrary lattices, and various lattice related problems like the Shortest-Integer-Solution,
Approximate Greatest Common Divisor.
As our main application of SNF we derive a new set of worst-to-average case lattice
reductions that deviate significantly from the template of Ajtai \cite{Ajt04}
and improve upon previous reductions in terms of simplicity.
\section{Introduction}

\subsection{The Systematic Normal Form : a "smoothed analysis" of lattices}

We introduce a certain canonical form of lattices called the systematic normal form (or SNF for short).
An SNF matrix is a matrix of this form:
\be
B = \left[\begin{array}{ccccc} 
N & b_{2} & b_{3} & \ldots & b_{n} \\
& 1 & & &\\
& & 1& &  \\
& & & \ddots & \\
& & & & 1
\end{array}
\right]
\ee
where all unspecified entries are $0$, and $N$ is a prime number.
The SNF form allows for a "smoothed analysis" of lattices in the following sense:
Given a matrix $B$ in Hermite normal form (HNF) one can efficiently generate a matrix $B_{SNF}$ in systematic normal form 
such that any solution of a problem of interest on $B_{SNF}$ can be translated easily to a solution for the original problem on $B$,
albeit with a slight degradation in accuracy.
Furthermore, $B_{SNF}$ is in a specialized canonical form with several important properties which are used in the reduction.
Hence, the reduction maps each lattice to a "nearby" lattice with the "right" properties.

The defining property of SNF, is that it maps the quotient structure $\Z^n / L$ as a quotient of infinite groups,
to a quotient  of finite groups $\F_N^n / L_N$, where $L_N = L \cap\F_N^n$ and $N = \det(B)$.
In addition, the very sparse structure of SNF allows for several other important properties among which we name:
\begin{enumerate}
\item
The dual to $L_N$ in $\F_N^n$, denoted by $(NL^*)_N$ also has a very structured matrix.
\be
B = \left[\begin{array}{ccccc} 
1 & &&& \\
-b_{2} & 1& & &\\
-b_{3} & & 1 & &  \\
\vdots & & & \ddots & \\
-b_{n} & & & & 1 
\end{array}
\right]
\ee
The duality between these two matrix forms has intriguing connections to well-known problems in cryptography:
On one hand, the primal SNF form is shown here to be synonymous with the problem finding short-integer-solutions $\sis$ (see Definition \ref{def:sis}, and connection in Proposition \ref{prop:equiv}),
whereas solving the closest vector problem on the dual lattice is very similar to the problem of approximate common divisor problem (\cite{HG01}). 
\item
It is straightforward to sample uniformly from both $L_N$ and $(NL^*)_N$:
to sample uniformly from $L_N$ one samples $n-1$ coordinates $x_n,\hdots, x_2$ uniformly at random from $\F_N$, and 
then sets $x_1$ as $x_1 = \sum_{i=2}^m b_i x_i \modn$.
For the dual lattice: one samples $a$ uniformly at random from $\F_N$, and then sets the dual vector as $(a,-b_2 \cdot a\modn, \hdots, -b_n \cdot a \modn)$.
\item
For every $x\in {\cal P}(L)$, i.e. the basic parallelotope of $L$, there exists a unique element in $(NL^*)_N$ that corresponds to $x$.
Furthermore, this element can be computed efficiently from $x$.
This is shown in Claim \ref{cl:dual}.
\item
The SNF form can be naturally associated with a rank-$1$ instance of $\sis$ modulo a prime number field:
Given a vector $g\in \F_N^m$ we define:
$$
L(g)_1^{\perp}:= \{ h\in \F_N^{m+1}, \langle (1,g), h \rangle =0 \in F_N\},
$$
One can then show that any SNF matrix $B$ has $L(B) = L(g)_1^{\perp}$ where $g$ is the vector of non-trivial entries in $B$, i.e. $b_2,\hdots, b_n$.
This is the subject of Proposition \ref{prop:equiv}.
\end{enumerate}

The formal connection between a given lattice basis $B$ and its nearby SNF matrix is given in Lemma \ref{lem:snf} which we sketch here:
\begin{lemma*}(sketch)
Let $B$ denote some integer matrix.
There exists an efficient algorithm that computes an SNF matrix $B_{SNF}$ and a number $T \leq 2^{\poly(n)}$ such that
$$\left\| B - \frac{1}{T} B_{SNF}\right\| = 2^{-\tilde{\Omega}(n)},$$ 
and for any $v\in L(B_{SNF})$, the vector $\hat{v} = B B_{SNF}^{-1} v$ has $\|v/T - \hat{v}\| = \OO(n^{-k})$,
for some fixed $k$.  
\end{lemma*}

\subsection{Worst-to-average case lattice reductions}

One of the most attractive features of lattice-based crypto-systems is the ability to reduce 
worst-case instances of such problems to average instances with a small loss in approximation error,
thus potentially paving the way to prove that they are one-way functions under standard complexity-theoretic assumptions,
namely ${\rm NP} \subsetneq {\rm BPP}$.

Usually, when speaking of worst-to-average case reductions on lattices, one thinks of a theorem
of the following form, as stated in the seminal work of Ajtai \cite{Ajt04}
\begin{theorem*}(Ajtai)
Let $A \in Z_q^{m \times n}$ be some random matrix, where $q, m = \poly(n)$, 
and consider $\Lambda(A)^{\perp} = \{x\in \F_N^n, A x = 0 (\mod \ q)\} \subseteq \Z^n$.
If ${\cal A}$ is an algorithm that finds a vector $v\in \Lambda(A)^{\perp}$, $\|v\| \leq \beta<< q$ with probability at least, say $1/n$, then 
there exists an algorithm ${\cal B}$ such that ${\cal B}^{\cal A}$ can solve $\gapsvp_{\gamma}$  
for $\gamma = \beta \poly(n)$, for any instance,
where $\svp_{\gamma}$ is the problem of finding a lattice vector which is at most $\gamma$ longer than the shortest one.
\end{theorem*}
The problem of finding a short vector in $\Lambda(A)^{\perp}$ is known as the Short-Integer-Solution or $\sis$,
and the statement above means that solving $\sis$ on average is at least as hard as approximating $\svp$
for suitable choice of parameters.

This line of research, pioneered by Ajtai \cite{Ajt04}, was later improved in the subsequent works of 
Micciancio and Regev \cite{MR07}, Gentry, Peikert and Vaikuntanathan \cite{GPV08}, 
and Micciancio and Peikert \cite{MP14}.
Specifically, in \cite{MR07}, the authors introduced the use of Gaussian measures, and harmonic (Fourier) analysis on the lattice
as a means to simply the reduction, and achieve better parameters.
Most prominently, they reduced the loss in accuracy: i.e. the ratio $\gamma/\beta$, to $\OO(n \sqrt{\log(n)})$, using a sequence of adaptive reductions.
 
Essentially, all prior reductions follow the same path, originally due to Ajtai.
At a very high level, they
amount to querying the $\sis$ oracle in order to iteratively improve (decrease the length) of a set of linearly independent
lattice vectors.
The oracle is applied to the vector of coefficients of the lattice vectors from previous rounds, 
which are then shown to be random (modulo $\Z_q^m$), so long as these vectors are of length which is comparable to the smoothing parameter of the lattice (see Definition \ref{def:smooth}).

Using the SNF structure, we propose a new worst-to-average case reduction that diverges from the template reduction due to Ajtai.
Specifically, we show a simplified scheme that reduces known lattices problem such as Guaranteed-Distance-Decoding ($\gdd$), and Shortest-Independent-Vector-Problem ($\sivp$) to
$\sis$.

Denote the smoothing parameter of a lattice by $\eta_{\eps}(L)$ (see Definition \ref{def:smooth}).
We consider the problem of Guaranteed-distance-decoding ($\gdd$) where we are given a lattice $L = L(B)$, and a target vector $v$, and asked to find a vector $x\in L$, such that $\dist(v,x) \leq \eta_{\eps}(L)$.
We also consider rank-$1$ instances of the Short-Integer-Solution problem as follows:
Fix some prime number $N$.
Given $g = (g_1,\hdots, g_n)$, 
where $g_i\in \F_N$
one can then define a lattice constrained by a "parity" check matrix,
as follows:
$$
\Lambda(g)^{\perp} = \{ h\in \F_N^n, \ \ \langle h, (g_1,\hdots, g_n) \rangle = 0 \in \F_N^n \}.
$$
The $\sis(N)$ problem (see Definition \ref{def:sis}), is then to find short vectors in $\Lambda(g)^{\perp}$, i.e. of length, say $\OO(n)$.
We note that while formally the ensemble above is a proper $\sis$ ensemble, 
for cryptographic applications
one usually uses a completely different range of parameters: the parity check matrix has rank $m = \poly(n)$ (and not $1$) and the field of interest is usually taken as $q = \poly(n)$, which is related to the encryption complexity.

We show the following theorem:
\begin{theorem*}(sketch)
Let $(B,v)$ be an input to $\gdd$
where $B$ is an $n\times n$ integer matrix.
Suppose that
$\eta_{\eps}(L) \leq \Phi$ for some $\eps = 2^{-n}$.
Suppose that ${\cal A}$ returns w.p. at least $1/\poly(n)$ a solution to
$\sis$.  
Then for $x_{out} = {\cal B}(B,v)$ we have that $x_{out}\in L$ and w.p. $\Omega(1)$
$$
\|x_{out} - v\| \leq \Phi \cdot n^{1.5} \cdot \max \{\log \det(B), n\}. 
$$
\end{theorem*}
A similar theorem (and algorithm) can be shown for reducing the problem of Shortest-Independent-Vector-Problem ($\sivp$) to $\sis$.
Following is a sketch of the algorithm:
\begin{algorithm*}(sketch)
\item
Input: A lattice $L = L(B)$ for some $n\times n$ HNF matrix $B$, and vector $v\in \F_N^n$.
\begin{enumerate}
\item 
Reduce $L$ to SNF form.
Denote matrix basis by $B_{SNF}$, matrix $M$, and constant $T$.
Denote $L_{SNF} = L(B_{SNF})$,  $N = \det(L_{SNF})$, $L_N = L_{SNF} \cap [N]^n$,  $NL^* = L(N \cdot B_{SNF}^{-T})$.
\item
Put 
$m$ as the minimal positive integer for which $\sqrt{m}^m \geq N$.
Set $s = T \cdot \eta_{\eps}(L)$ for some $\eps = \OO(1)$.
\item
Choose $c$ uniformly at random from $\F_N \cap [- m\sqrt{m}, m\sqrt{m}]$.
Choose $u_{rand}\in L_N$ uniformly at random from $L_N$.
Put $v_{target} = (c^{-1}\cdot T \cdot v + u_{rand}) \modn$.
\item
Repeat $m$ times:
\begin{enumerate}
\item
Sample $x_i = (x_{i,1},\hdots, x_{i,n})\sim \rho_{\F_N^n,s,v_{target}}$.
\item
Compute $y_i$ as the point in $NL^*$ corresponding to $x_i$.
Denote by $a_i$ the first coordinate of $y_i$.
\end{enumerate}
\item
Put $\{\alpha_i\}_{i=1}^{m} = {\cal A}(a_1,\hdots, a_m)$.
\item
If ${\cal A}$ fails or $\sum_{i=1}^m \alpha_i \neq c$ return FAIL.
\item
Compute $x_0 = (\sum_{i=1}^m \alpha_i (y_i + x_i) - c u_{rand})\modn$.
\item
Return $x_{out} := B (B_{SNF}^{-1}M^{-1} \cdot x_0) $.
\end{enumerate}
\end{algorithm*}

The algorithm, given target vector $t$, tries to find a "relatively" close lattice vector to $t$.
This is, in a sense, the same problem as $\cvp$ (the closest-vector problem), except here we are only interested in an approximation which
is comparable to the "smoothing parameter" of the lattice (see Definitions \ref{def:gdd} and \ref{def:cvp} for a formal statement).
Using Lemma \ref{lem:snf}, which shows how to reduce an arbitrary lattice to SNF form, and then translate
back the solution to the original lattice, we will describe the algorithm assuming the input lattice is already in SNF form.

The reduction begins by sampling from a discrete Gaussian around the target vector $t$, with sufficiently large variance.
Denote these samples by $x_1,\hdots, x_n\in \F_N^n$.
By the SNF structure, the algorithm then computes for each sampled point $x_i$, a corresponding point $y_i\in (NL^*)_N$.
such that
$$
\forall i \ \ x_i + y_i \in L_N. \
$$
We do not know whether such a bijection exists for general lattices. 

The structure of the dual lattice implies that the points $y_i$ are completely characterized by their first coordinate.
We then show that for sufficiently large Gaussian variance, the first coordinate of these points
is uniformly distributed on $\F_N^n$.
This implies that, given as input the first coordinate $a_i$ of each $y_i$, the random oracle for SIS will succeed with high probability.
Suppose that the oracle to $\sis$ returns coefficients $\{\alpha_i\}_{i=1}^m$
such that $\sum_i \alpha_i a_i = 0 \in \F_N$, and consider the linear combination of the vectors $x_i$ using these coefficients:
$$
x_0 = \sum_{i=1}^m \alpha_i (x_i+y_i)
$$
We note that the above is a lattice point in $L_N$ because it is an integer combination of vectors in $L_N$.
By the SNF structure this implies that $\sum_i \alpha_i y_i = 0$.
Hence, the distance of $x_0$ to the target $c$ is given by the vector:
$$
\sum_{i=1}^m \alpha_i x_i = \sum_i \alpha_i t + \sum_{i=1}^m \alpha_i {\cal E}_i,
$$
where ${\cal E}_i$ is a discrete Gaussian on $\F_N^n$ centered around $0$ with the same variance as $x_i$.
Since the $\alpha_i$'s are small (say $\sqrt{n}$) and the variance of ${\cal E}_i$ is $s^2$, 
then to upper bound the length of $x_0-t$ it is sufficient to make sure that the linear combination of $\alpha_i$'s is in fact {\it affine}, i.e.
$\sum_{i=1}^m \alpha_i = 1$.

It is not immediately clear why a short-integer-solution to a random instance should be affine.
We do know however, that although it may not be affine, the sum of coefficients is quite small, say $\OO(n\sqrt{n})$.
Thus, 
a natural scheme would be to "guess"
the sum of coefficients $c = \sum_{i=1}^m \alpha_i$ in advance out of a small interval, and then sample $x_i$'s from 
a Gaussian centered around a scaled target - $\rho_{\F_N^n, s, c^{-1} t}$.
For such a scheme to work, we need to rule out the possibility that since we change the original samples to be centered around $c^{-1} t$, the oracle to $\sis$, despite returning a good answer for almost all inputs, will "adversarially" pick
an answer $(\alpha_1,\hdots, \alpha_m)$ such that $\sum_{i=1}^m \alpha_i \neq c$ {\it always}.

To rule out this possibility, we introduce extra randomization to make sure that the random oracle has no knowledge of which target vector we are looking for.
We add a random lattice vector to $t$, for which we compensate at the end, and then argue that the random oracle can only gain information about the value of $t$ modulo the basic parallelotope ${\cal P}(L)$.  But since the $x_i$'s are sampled with variance above the smoothing parameter, it implies that modulo ${\cal P}(L)$ the values of the $x_i$'s are almost completely uniform.

\subsection{Discussion and prior art}

The above theorem improves on prior art, and specifically \cite{MR07} in the following ways:
\begin{enumerate}
\item
The reduction is direct:  in \cite{MR07} the reduction from $\gdd$ to $\sis$ goes through the problem of finding a short linearly independent set called
${\rm SIVP}$, and through a variant of $\gdd$ called ${\rm INC-GDD}$.
Hence our reduction is somewhat simpler to follow.
\item
It does not rely on the Ajtai form of reductions.
Hence it allows to consider other problems as hard random instances: for example, 
a linear congruence modulo a large prime number $N$, instead of simultaneous linear equations modulo a small field $q$
\cite{Ajt04}.
\item
It is inherently non-adaptive: the oracle calls to $\sis$ can be made once in parallel, following which the algorithm returns
an answer to the original problem ($\gdd$ or $\sivp$).
\end{enumerate}

By \cite{AR04} it is known that approximating $\svp,\cvp$ to a factor at most $\sqrt{n}$
is in ${\rm NP} \cap {\rm coNP}$, and by \cite{GG} and \cite{GMR} a $\sqrt{n/\log(n)}$ approximation is in ${\rm NP} \cap {\rm coAM}$.
Hence, much larger improvement is needed in order to approach the domain of NP-hardness.
In addition, since the reduction is non-adaptive, then by \cite{BT} it cannot be used to 
show a reduction from an NP complete problem to an average case problem, unless the polynomial hierarchy collapses.
This suggests, that perhaps an adaptive variant of our algorithm can result a further improvement of the approximation factor.

Finally, we point out that it is not known \cite{RS} whether $\gdd_p$ is as hard as $\svp_q$ for any polynomials $p,q$, hence
it is possible that approximate $\gdd$ (or $\sivp$) is in fact an {\it easier} problem than approximate $\svp,\cvp$.

\section{Preliminaries}

\subsection{Notation}

The $n$-dimensional Euclidean space is denoted by $\R^n$.
The Euclidean norm of a vector $x\in \R^n$ is $\|x\| = \sqrt{\sum_{i=1}^n |x_i|^2}$.
A Euclidean lattice $L$ is written as $L = L(B)$ where $B$ is some basis of $L$.
$N$ is used to denote a prime number, and $\F_N$ the prime number field corresponding to $N$.
We define $\Delta$ as the statistical distance between distributions $(p,\Omega), (q,\Omega)$,  i.e. $\Delta(p,q) = \int_{\Omega} |p(x) - q(x)| dx$.
Given a set $S$, $U(S)$ is the uniform distribution on $S$.
For any $v\in \mathbf{R}^n$ define: 
$
|v| = max_i |v_i|.
$ 
For real number $s>0$ and vector $c\in R^n$, $\bar{B}_s(c)$ is the closed Euclidean ball of radius $s$ around $c$.
For integer $n\geq 1$, the notation $[n]$ stands for the set of indices $\{1,\hdots, n\}$.
Given a set $S\subseteq \R^n$, and a vector $v\in \R^n$, we denote $\dist(v,S): = \min_{x\in S} \|v - x\|$.
For a positive integer $M$, we denote by $[M]$ as the interval of integers $[0,\hdots, M-1]$.

We say that a problem $P$ is efficiently computable if there exists an algorithm that runs in time $\poly(n)$,
where $n$ is the size of the description of an instance of $P$.

\subsection{Background on lattices}
We start by stating some standard facts about lattices.

\begin{definition}

\textbf{Euclidean Lattice}

\noindent
A Euclidean lattice $L\subseteq \R^n$ is the set of all integer linear combinations of a set of linearly independent vectors
$b_1,\hdots, b_m$:
$$
L = \left\{ \sum_{i=1}^m z_i b_i, \ \ z_i\in \Z, \right\} \subseteq \R^n
$$
This set $\{b_i\}_{i=1}^n$ is called the {\em basis} of the lattice.
We denote by $L = L(B)$, where $B$ is the matrix whose columns are $b_1,\hdots, b_m$.
In this paper, we will always assume that $L$ is full-dimensional, i.e. $m=n$.
\end{definition}
\noindent
For lattice $L = L(B)$, ${\cal P}(B)$ is the basic parallelotope of $L$ according to ${\cal B}$:
$$
{\cal P}(B) := \left\{ v= \sum_{i\in [n]} x_i b_i, \ \ x_i\in [0,1) \right\}.
$$
Sometimes, it will be more convenient to use ${\cal P}(L)$ which is independent of the basis.
This then refers to ${\cal P}(B)$ for some basis $B$ for $L$.

%\begin{definition}
%
%\textbf{Hermite Normal Form}
%
%\noindent
%A full-rank square matrix $B$ is said to be in Hermite Normal Form (HNF) if
%\begin{enumerate}
%\item
%$b_{i,j}=0$ for all $i>j$.
%\item
%$0 \leq b_{i,j} < b_{i,i}$ for $i<j$.
%\end{enumerate}
%\end{definition}
%The Hermite normal form of an integer lattice can be computed in polynomial
%time for any lattice.   
%
%For $L = L(B)$ where $B$ is HNF, we define the determinant of the lattice as
%\be 
%N = \prod_{i=1}^n b_{i,i} = \det B. 
%\ee
\begin{definition}

\textbf{The Dual Lattice}

\noindent
The dual of a lattice is the lattice generated by
the columns of $B^{-T}$. 
\end{definition}
%It is easy to see that the matrix $B^\perp$ is in lower triangular form. In
%particular, the dual of our lattice in Hermite normal form is the column
%space of 
%a matrix $B^\perp$ that looks
%like 
%\be
%B^\perp = \left[\begin{array}{ccccc} 
%b_{1,1}^{-1} & & & & \\
%* & b_{2,2}^{-1} & & & \\
%* & * & b_{3,3}^{-1} & & \\
%\vdots & & & \ddots & \\
%* & * &*  &\ldots  & b_{n,n}^{-1} 
%\end{array}
%\right]
%\ee
%where $*$ represents a possibly non-zero entry and the blanks represent
%zeroes. Every entry of $B^\perp$ is a fraction with denominator $D$. 
%If we define $u_i$ and $v_j$ to be the $i$th and $j$th columns of $B$ and
%$B^\perp$, we have $\langle u_i,v_j \ra = \delta_{i,j}$, because 
%$(B^\perp)^T B = (B^T B)^{-1} B^T B = I$. 

\begin{definition}

\textbf{Successive minima of a lattice}

\noindent
Given a lattice $L$ of rank $n$, its successive minima $\lambda_i(L)$ for all $i\in [n]$ are defined as follows:
$$
\lambda_i(L) = 
\inf 
\left\{ 
r | \dim( {\rm span} (L \cap \bar{B}_r(0)) ) \geq i 
\right\}.
$$
\end{definition}

\begin{definition}

\textbf{Unimodular matrix}

\noindent
The group of unimodular matrices $GL_n(\Z)$ is the set of $n\times n$
integer matrices with determinant $1$.
Unimodular matrices preserve a lattice: $L(B) = L(B')$ if and only if $B = B' \cdot A$, for some unimodular matrix $A$.
\end{definition}

\begin{definition}

\textbf{The determinant of a lattice}

\noindent
For a lattice $L = L(B)$ we define $\det(L) = \det(B)$, and denote by $N$.
\end{definition}
The determinant of a lattice is well-defined, since if $L(B') = L(B)$, then by the above $B = B' \cdot A$
for some unimodular matrix $A$, in which case $\det(B') = \det(B) \det(A) = \det(B)$.

The lattice $L$ is periodic modulo $N$. In other words, if we add $N$ to 
any coordinate of a lattice point, we reach another
lattice point. Thus, a cube of side length $N$ gives a subset of the lattice
which generates the whole lattice when acted on by translations by $N$ in
any direction. We let $L_N$ denote the lattice restricted to a cube of side
length $N$.  

In particular, if $L = L(B)$ is an integer lattice, with $\det(L) = N$, for prime $N$, this implies that $L_N$ is a lattice of $\F_N^n$:
\begin{proposition}
Let $\F_N^n$ denote the additive group of $n$-dimensional vectors of integers, where in each coordinate summation is carried out modulo $N$.
Then $L_N$ is an additive sub-group of $\F_N^n$, that contains the $0$ point.
In particular $L_N$ is a lattice of $\F_N^n$.
\end{proposition}

\noindent
A canonical representation of integer lattices is called the Hermite normal form (HNF):
\begin{definition}

\textbf{Hermite Normal Form}

\noindent
An integer matrix $A\in \Z^{n\times n}$ is said to be in Hermite normal form (HNF) if $A$ is 
upper-triangular, and $a_{i,i}> a_{i,j}\geq 0$ for all $j>i$, and all $i\in [n]$.
\end{definition}

\noindent
It is well-known that every integer matrix can be efficiently transformed into HNF:
\begin{fact}

\textbf{Unique, efficiently-computable, Hermite normal form}

\noindent
For every full-rank integer matrix $A\in Z^{n\times n}$, there exists a unique unimodular
matrix $U\in GL_n(\Z)$, such that $H = U \cdot A$, and $H$ is HNF.
$U$ can be computed efficiently.
\end{fact}

\subsection{Gaussians on lattices}

The use of the Gaussian measure in the context of lattices is well-known in math.
In the context of lattices, the use of the Gaussian measure has been shown in recent years to be very useful to derive 
important geometric facts about lattices \cite{Ban93}
complexity-theoretic results \cite{AR04},
the well-known Learning-with-Errors public-key crypto-system \cite{Reg09},
hard random lattices \cite{MR07} and the fastest classical algorithms for
the shortest vector problem \cite{ADR+15}.

\begin{definition}

\textbf{The discrete Gaussian measure over lattices}

\noindent
For any $s > 0$ define the Gaussian function on $\R^n$ centered at $c$ with parameter $s$:
$$
\forall x\in R^n , \ \rho_{s,c}(x)  = \exp( - \pi \|x - c\|^2 / s^2 ).
$$
For any $c\in \R^n$, real number $s>0$, and $n$-dimensional lattice $L$,
the discrete Gaussian distribution over $L$ is:
$$
\forall x\in L, \ \ D_{s,c}(x) = \frac{\rho_{s,c}(x)}{\rho_{s,c}(L)},
$$
where
$$
\rho_{s,c}(L) = \sum_{x\in L} \rho_{s,c}(x).
$$
\end{definition}

\noindent
In \cite{MR07}
Micciancio and Regev introduced a lattice quantity called the smoothing parameter:
\begin{definition}[MR07]\label{def:smooth}

\textbf{The smoothing parameter of a lattice}

\noindent
For any $n$-dimensional lattice and positive real $\eps>0$, the smoothing parameter
$\eta_{\eps}(L)$ is the smallest real number $s>0$ for which
$\rho_{1/s}(L^* - \{0\}) \leq \eps$.
\end{definition}
\noindent
We mention several important facts by Micciancio, Regev \cite{MR07}.
The first one can be regarded, in a sense
as the defining property of the smoothing parameter:
\begin{fact}\label{fact:mr}
Let $L$ be some integer lattice $L\subseteq \R^n$, $\eps\in (0,1)$.  If $s\geq \eta_{\eps}(L)$ then for all $c\in \F_N^n$
$$
\Delta \left(D_{\F_N^n,s,c} (\mod {\cal P}(L)), U({\cal P}(L)) \right)\leq \eps.
$$
\end{fact}
\noindent
The second one states that the smoothing parameter can be chosen arbitrarily close to the $n$-th minima of the lattice:
\begin{fact}\label{fact:nmin}
For any lattice $L$, and $\eps = n^{-k}$ for some constant $k$, we have $\eta_{\eps}(L) = \OO(\lambda_n(L) \cdot \sqrt{\log(n)})$.
\end{fact}
\noindent
The third one - is that if we sample from the discrete Gaussian on a lattice $L$ with variance $s^2$, then typically
a returned vector will have length $s \sqrt{n}$.
\begin{fact}\label{fact:ban}
For any $n$-dimensional lattice $L$, $\eps = o(1)$, real vector $c\in \R^n$, and real number $s \geq \eta_{\eps}(L)$, we have 
$$
D_{L,s,c}(L - B_{s}(c)) = 2^{-\Omega(n)}.
$$
\end{fact}

\subsection{Lattice problems}

\begin{definition}

\textbf{Closest-vector problem / Shortest-vector problem}

\noindent
The closest-vector problem is defined as follows:
Given is a lattice $L = L(B)$, and a vector $v\in \mathbf{R}^n$.
Find a lattice vector $w$ for which $\|v - w\| = \dist(v,L)$.
The shortest-vector problem is defined as follows:
Given $L = L(B)$ find a non-zero lattice vector of minimal length.
\end{definition}

\begin{definition}\label{def:cvp}

\textbf{Approximate closest-vector (search problem)}

\noindent
The approximate shortest-vector problem $\cvp_{\beta}$  is the following problem:
given a lattice $L$, and a vector $v\in \mathbf{R}^n$
return $w\in L$ such that $\dist(v,w) \leq \beta \cdot \dist(v,L)$.
\end{definition}

%\begin{definition}
%
%\textbf{Approximate closest-vector (decision problem)}
%
%\noindent
%The approximate closest-vector problem $\gapcvp_{\beta}$  is the following problem:
%given a lattice $L$, a vector $v\in \mathbf{R}^n$, and a parameter $\gamma$
%decide among two cases, with the promise that the input belongs to one of them:
%\item
%YES: $\dist(v,L) \leq \beta$
%\item
%NO: $\dist(v,L) \geq \beta \gamma$
%\end{definition}

\begin{definition}\label{def:gdd}

\textbf{Guaranteed-distance decoding - $\gdd$}

\noindent
We are given an $n$-dimensional integer lattice $L = L(B)$, and $v\in \Z^n$ be some vector.
Fix $\eps = \poly(n)$.
The problem $\gdd$ is to find a vector $w\in L$, such that $\dist(w,v) \leq \eta_{\eps}(L)$.
\end{definition}
\noindent
The approximate version $\gdd_{\gamma}$ is to find such a vector $w\in L$ with  $\dist(w,v) \leq \eta_{\eps}(L) \cdot \gamma$.

\begin{definition}\label{def:sivp}

\textbf{Shortest Independent Vectors Problem}

\noindent
Let $L = L(B)$ be some $n$-dimensional lattice.
The problem $\sivp$ is to find a set of linearly independent vectors in $L$, of maximal length at most $\lambda_n(L)$.
\end{definition}
\noindent
The approximate version $\sivp_{\gamma}$ is to find such vectors whose length is at most $\lambda_n(L) \cdot \gamma$.

\section{The Systematic Normal Form (SNF)}

\begin{definition}\label{def:snf}

\textbf{Systematic Normal Form (SNF)}

\noindent
A matrix $B$ is said to be SNF if $B_{i,i}=1$ for all $i>1$,
and $B_{1,1}=N$ where $N$ is a prime number, and in addition, for all $i>1$ $B_{i,j}=0$ for all $i\neq j$.
\end{definition} 
This form is called suggestively "systematic" because for every $v\in L(B)$, the last $n-1$ coordinates,
are in fact the last $n-1$ coefficients of the vector under the basis $B$, which in error-correcting terminology can be considered as the "message" to be encoded by the matrix $B$.

\noindent
The following facts will be useful later on:
\begin{proposition}
If $B$ is in SNF form, then $NB^{-T}$, i.e. the matrix spanning the scaled dual of $L(B)$ assumes the following form:
\be\label{eq:ndual}
B = \left[\begin{array}{ccccc} 
1 & &&& \\
-b_{2} & 1& & &\\
-b_{3} & & 1 & &  \\
\vdots & & & \ddots & \\
-b_{n} & & & & 1 
\end{array}
\right]
\ee
\end{proposition}
\begin{proposition}
There are $N^{n-1}$ points of $L = L(B_{SNF})$ in $\F_N^n$, and there are $N$ lattice points of $NL^*$ in that cube.
Hence there are $N$ points of $\F_N^n$ inside ${\cal P}(L)$, and $N^{n-1}$ points in ${\cal P}(NL^*)$. 
\end{proposition}

In fact, a somewhat stronger statement is true:
\begin{proposition}\label{prop:bij}
There exist bijections $\Phi_1: \F_N^n / L_N \mapsto {\cal P}(L)$,
$\Phi_2: {\cal P}(L) \mapsto NL^*$ as follows:
For every coset of $L_N$ in $\F_N^n$ $\Phi_1$ returns a unique element of ${\cal P}(L)$, 
and $\Phi_2$ maps each element of ${\cal P}(L)$ uniquely to an element of $(NL^*)_N = NL^* \cap \F_N^n$ in that coset.
Thus, in particular for every $x\in {\cal P}(L)$, there exists a unique $z\in (NL^*)_N$, such that $x + z\in L_N$.
\end{proposition}

Perhaps more interestingly, though, for SNF one can compute efficiently, for each integer vector $x$, a corresponding
dual-lattice vector $y$, such that their sum is in $L_N$:
\begin{claim}\label{cl:dual}

\textbf{Compute dual vector for any vector}

\noindent
For $x\in \F_N^n$, the map $\Phi_3:= \Phi_2 \circ \Phi_1(x)$ 
can be computed efficiently.
\end{claim}
\begin{proof}
Let $x\in \F_N^n$.
We would like to find (the unique) $y = \Phi_3(x) =  \Phi_2 \circ \Phi_1(x) \in NL^*$ for which $x + y\in L_N$.
Each point in $y\in NL^*$ is characterized uniquely by an element $a\in \F_N$ as follows:
\be\label{eq:nl}
y = (a, -b_2 a  \modn,\hdots, -b_n a\modn).
\ee
Thus, to find $y$ we would like to solve the following vector equality over $a,z_2,\hdots, z_n\in \F_N$:
\be
(x_1,\hdots, x_n)^T + (a, -b_2 a \modn,\hdots, -b_n a \modn)
=
\left(\sum_{i=2}^n b_i z_i \modn, z_2,\hdots, z_n\right)^T
\ee
Consider the first coordinate. We have:
\be
x_1 + a = \sum_{i=2}^n b_i z_i \modn.
\ee
Substituting in the above $z_i = x_i - a b_i \modn$ for all $i\geq 2$ implies:
\be
x_1 - \sum_{i=2}^n x_i b_i =  - a \cdot \left( \sum_{i=2}^n b_i^2 + 1 \right) \modn.
\ee
Since $N$ is prime, then the number
$\sum_{i=2}^n b_i^2 + 1$ has an inverse.
Thus, the parameter $a$ can be
computed uniquely from the equation above, which implies that $y$ can be determined uniquely and efficiently.
\end{proof}

\subsection{Reduction to SNF form}

In this section we provide an efficient reduction from an arbitrary lattice to a lattice in SNF form,
that preserves all important properties of the lattice.
Specifically - 
it allows the reduction of any computational problem on an arbitrary lattice $L$ to another problem on an SNF lattice $L_{SNF}$
such that any solution to the reduced problems allows to find {\it efficiently} a solution to the problem on $L$.

\begin{lemma}\label{lem:snf}

\textbf{Efficient reduction to SNF}

\noindent
There exists an efficient algorithm that for any LLL-reduced upper-triangular matrix $B$, 
and numbers $a>0,b>0$, 
computes efficiently a tuple $\langle B_{SNF},M,T \rangle$, where
$B_{SNF}$ is an SNF matrix, $T$ is a positive integer
$T = \max\{ 2^{\OO(n)}/ \lambda_1(L(B)),\det(B) \}$, and $M\in GL_n(\Z)$, such that the following holds:
For any $v\in L(B_{snf}/T)$ put $\hat{v} = B \cdot (T B_{snf}^{-1} M v)$.
If $\|v\| \leq \det(B) \cdot n^a$
then $\hat{v}\in L(B)$ and  $\left\|\hat{v} - v\right\|  = \OO(n^{-b})$.
Also $\det(B_{SNF}) = \OO(\det(B) \cdot T^n)$.
\end{lemma}

\noindent
Before presenting the proof,
let us bound the coefficients of any short vector in a lattice. 
%\begin{proposition}
%Suppose that we have a $n$-dimensional lattice $L$ with basis vectors $v_1$, $v_2$,
%$\ldots$, $v_n$. Let $E_{max}$ be the largest absolute value of any
%entry in any of the basis vectors. 
%Then, any lattice vector $v$ can be represented in 
%this basis with the coefficient on $v_i$ at most $E_{max}^{n-1} n^n |v|/ \det(L)$.
%\end{proposition}
%
%\begin{proof} We use Cramer's rule. Cramer's rule says that if $A$ is an $n \times n$
%matrix, and $A \ell = v$, then the
%$j$th coordinate of $\ell$ is 
%\be
%\frac{\det \left( c_1 c_2 \ldots c_{j-1} v c_{j+1} \ldots c_n \right) }{\det A}
%\ee
%where $c_i$ is the $i$th column of $A$. 
%
%We can apply Cramer's formula to a vector in our
%lattice. The numerator of Cramer's formula for the coefficient
%on $v_j$ is a matrix where all entries not in the $j$th column are 
%at most $E_{max}$, and where the $j$th column 
%has all entries at most $|v|$. We
%bound the determinant using the Leibnitz formula for the determinant:
%\be
%\det A = \sum_{\sigma \in S_n} \mbox{sign}(\sigma) \prod_{i=1}^n A(i,\sigma_i).
%\ee
%There are $n!$ terms in this formula, and each term is at most $E_{max}^{n-1} |v|$. 
%Thus, the $j$'th coordinate is at most $E_{max}^{n-1} n^n |v|/\det(L)$. 
%\end{proof}

\begin{proposition}\label{prop:lll1}
Let $B$ be an $LLL$-reduced matrix, and $v\in L$ be some lattice vector.
Then $v$ can be represented in the basis $B$ using a vector coefficients of length at most $ \|v\|\cdot 2^{\OO(n)}/ \lambda_1(L(B))$.
\end{proposition}
\begin{proof}
Write $v = Bx$, where $x\in \Z^n$.
Consider the QR decomposition of $B$ as $B = Q \cdot R$, where $Q$ is a unitary matrix,
and $R$ is an upper-triangular matrix.
Thus $Q\cdot R$ correspond to the Gram-Schmidt decomposition of $B$.
We express $x$ as:
\be
\|x\| = \|B^{-1} v\| \leq \sqrt{n} R_{n,n}^{-1} \|v\|,
\ee
where the second inequality follows from unitarity of $Q$.
By LLL we know that for some constant $\beta>0$ we have:
\be
\lambda_1:= \lambda_1(L) \leq R_{1,1} \leq 2^{\beta n} \cdot R_{n,n}.
\ee
Therefore
\be
R_{n,n} \geq 2^{- \beta n} \lambda_1.
\ee
This implies that
\be
\|x\| \leq \sqrt{n} R_{n,n}^{-1} \|v\| \leq \|v\| \sqrt{n} \cdot 2^n/\lambda_1  = \|v\| \cdot 2^{\OO(n)} / \lambda_1.
\ee
\end{proof}

\noindent
The following are easy corollaries of the above:

%\begin{proposition}
%Let $B$ be an $LLL$-reduced integer matrix, 
%and $B'$ be some invertible matrix for which $\|B - B'\|  = 2^{-\OO(n)}$,
%and $v\in L(B')$ be some lattice vector.
%Then $v$ can be represented in the basis $B$ using a vector coefficients of length at most $2^{\OO(n)} \|v\|$.
%\end{proposition}

\begin{proposition}\label{prop:close}
\label{two-close-lattices}
Let $B_1 = \{v_i\}_{i=1}^n$ be some LLL-reduced basis
and another basis $B_2 = \{w_i\}_{i=1}^n$ for lattice $L_2 = L(B_2)$.
Suppose that $\|v_i - w_i\| \leq \alpha$. 
Let $v = \sum_{i=1}^n c_i v_i$ be a point in $L_1$ and $w = \sum_{i=1}^n
c_i w_i$ be the corresponding point in $L_2$. Then $\left\|v - w \right\| \leq \|v\| \alpha 2^{\OO(n)}/\lambda_1(L_1) $. 
\end{proposition}

\begin{proof}
By the triangle inequality we have:
\begin{eqnarray}
\left\|v-w \right\| &= &  \left\|\sum_{i=1}^n c_i v_i  -  \sum_{i=1}^n c_i w_i  \right\| \\
&\leq & \sum_{i=1}^n \| v_i -w_i \| |c_i| \\
&\leq & n\alpha \|v\| \cdot 2^{\OO(n)} /\lambda_1(L_1)\\
&=     & \alpha \|v\| \cdot 2^{\OO(n)}/\lambda_1(L_1),
\end{eqnarray}
where the inequality before last follows from Proposition \ref{prop:lll1}. 
\end{proof}

\subsection{Proof of Lemma \ref{lem:snf}}

\begin{proof}
We first use $T$ as a parameter and determine it later on in the proof.
We start from an upper-triangular LLL-reduced matrix $B_1$:
\be
B_1 = \left[\begin{array}{ccccc} 
b_{1,1} & b_{1,2} & b_{1,3} & \ldots & b_{1,n} \\
& b_{2,2} & b_{2,3} & \ldots & b_{2,n} \\
& & b_{3,3} & \ldots & b_{3,n} \\
& & & \ddots & \vdots \\
& & & & b_{n,n} 
\end{array}
\right]
\ee
add $1/{T}$ along the subdiagonal,
and truncate each non-zero entry to its nearest integer multiple of $1/T$:
\be
\label{add-subdiagonal}
B_2 = \left[\begin{array}{ccccc} 
b_{1,1}' & b_{1,2}' & b_{1,3}' & \ldots & b_{1,n}' \\
\frac{1}{T}& b_{2,2}' & b_{2,3}' & \ldots & b_{2,n}' \\
& \frac{1}{T}& b_{3,3}' & \ldots & b_{3,n}' \\
& & \ddots & \ddots & \vdots \\
& & & \frac{1}{T}& b_{n,n}'
\end{array}
\right],
\ee
where $b_{i,j}' = \round{b_{i,j} T} /T$.
We note that
\be\label{eq:entry}
\forall i,j \ \ |B_2(i,j) - B_1(i,j)| \leq 1/T.
\ee

We now use column operations to make rows $2$, $3$, $\ldots$,
$n$ of the lattice zero except for the subdiagonal. This involves subtracting integer multiples
of the $i$th column from all later columns. We obtain a lattice of the form.
\be\label{eq:b3}
\label{after-Gaussian}
B_3 = \left[\begin{array}{ccccc} 
b_{1,1}' & b_{1,2}'' & b_{1,3}'' & \ldots & b_{1,n}'' \\
\frac{1}{T}& 0 & 0 & \ldots & 0 \\
& \frac{1}{T}& 0 & \ldots & 0 \\
& & \ddots & \ddots & \vdots \\
& & & \frac{1}{T}& 0 
\end{array}
\right]
\ee
Observe that if we move the $n$th column to the first column and multiply all entries by
$T$, we now have a lattice which is a $1/T$ multiple of an SNF lattice, 
except possibly from the entry $b_{1,n}''$ which may not be prime.

We now show how to make determinant of the new lattice prime,
by rounding
$b_{1,n}''$ to a prime.  
By standard number-theory conjecture
\footnote{Generalized Riemann Hypothesis (GRH)} we assume:
\be\label{eq:delta}
\exists \delta = \OO(\log(T b_{1,n}'')), \ \ T b_{1,n}'' + \delta \mbox{ is prime }.
\ee

We need to compute the matrix that transforms the basis given in equation
(\ref{add-subdiagonal}) to the basis given in equation (\ref{after-Gaussian}).
That is, we want the matrix $M$ such that $B_3 =  B_2 M$. The diagonal and
superdiagonal of the matrix can be easily calculated: 
\begin{eqnarray}
M &=& \left( 
\begin{array}{ccccccc}
1 & T b_{2,2}' & T b_{2,3}' &  \ldots & T b_{2,n-1}'&   T b_{2,n}' \\
  & 1 & T b_{3,3}' & \ldots & T b_{3,n-1}' & T b_{3,n}' \\
  &   & 1 & \ldots & T b_{4,n-1}' & T b_{4,n}' \\
  &      &   & \ddots &   \vdots & \vdots \\
  &   & & & 1 & T b_{n,n}'  \\ 
  &   &  & &   & 1 
\end{array}
\right)^{-1} \\
&=&
\left( 
\begin{array}{ccccccc}
1 & - T b_{2,2}' &   \ldots &   \\
  & 1 & - T b_{3,3}' & \ldots &  \\
  &   & 1 & - T b_{4,4}' & \ldots\\
  &      &   & \ddots &   \ddots  \\
  &   & & & 1 & - T b_{n,n}'  \\ 
  &   &  & &   & 1 
\end{array}
\right)
\end{eqnarray}
By the above, $M$ is a unimodular matrix, hence $\det(B_2) = \det(B_3)$.
Note that both $M$ and $M^{-1}$ are upper triangular matrices with $1$s along the diagonal.
The determinant of $TB_3$ is $T b_{1,n}''$. 
To obtain a lattice with a prime determinant, we need to round $\det(T B_3)$ to a 
nearby prime.  Let us assume that $\det{T B_3} + \delta$ is prime. This rounding 
corresponds to adding $\delta/T$ to the entry $B_3(1,n)$. 

What effect does this change have on the basis of the lattice in $B_2$?
Let $\Delta$ be the matrix with
$\Delta(1,n) = \delta/T$ and all other entries $0$. Then our matrix in the SNF basis
is $B_3 + \Delta$. To see what the effect on $B_2$ is, we merely need to multiply by
$M^{-1}$. That is,
\be
B_2 + \Delta M^{-1} = ( B_3 + \Delta )M^{-1} .
\ee
Using the form we derived above for $M^{-1}$, we see that because there are 1s along the
diagonal of $M$ then 
$\Delta M^{-1} = \Delta$. Thus, we can make $B_2$ obtain a prime determinant by simply
adding $\delta/T$ to $B_2(1,n) = b_{1,n}'$. This
changes the length of the $n$th basis vector by at most
$\delta/T$.

Let $B_4$ denote then the output SNF matrix.
\be
B_4 = \left[\begin{array}{ccccc} 
T b_{1,n}'' + \delta & T b_{1,1}' & T b_{1,2}'' & \ldots & T b_{1,n-1}'' \\
0 & 1 & 0 & \ldots & 0 \\
   & 0 & 1 & \ldots & 0 \\
& & \ddots & \ddots & \vdots \\
& & & 0& 1 
\end{array}
\right]
\ee
By Equations \ref{eq:entry} and equation \ref{eq:delta} :
\be
\forall i,j\in [n] |(M^{-1} B_4(i,j))/T - B_1| = \OO(\log(T b_{1,n}'')/T).
\ee
That is, the basis $M^{-1} B_4/T$ of $L(B_4)/T$ is entry-wise close to $B_1$.
By assumption $B_1$ is in particular LLL-reduced.
Hence we can invoke Proposition \ref{prop:close} w.r.t. these two bases:
Consider some $v\in L(B_4/T)$.
Let $\beta$ denote the implicit constant in the bound of Proposition \ref{prop:close}.
Applying Proposition \ref{prop:close} implies that 
the 
corresponding vector $\hat{v} = B_1 \cdot (TB_4^{-1}M) \cdot v\in L(B_1)$
has
\be\label{eq:bound2}
\left\|\hat{v}- v\right\| \leq \frac{2^{\beta n} \|v\| \log(T b_{1,n}'')}{T\lambda_1(L(B_1))}.
\ee
By assumption $\|v\| \leq n^a \det(B_1)$.
Then together with Equation \ref{eq:bound2} we conclude that there exists
\be
T = \poly(n) \cdot \max \{ \det(B_1)/\lambda_1(L(B_1)), 2^{\beta n} \}.
\ee
such that
\be
\left\|\hat{v}- v\right\| \leq  n^{-b}.
\ee
Finally, the entry $B_{SNF}(1,1) = \det(B_{SNF}) \leq 2 b_{1,n}'' \cdot T^n$, hence the determinant $\det(B_{SNF})$ is upper-bounded by 
\be
\det(B_{SNF})  = \OO(\det(B) T^n).
\ee

\end{proof}

The lemma above implies that one can reduce the standard lattice problems, given for an arbitrary lattice, to the same problem on a lattice in SNF,
and then translate the output solution efficiently to a solution for the original lattice:
\begin{corollary}\label{cor:cvp}

\textbf{SNF reduction preserves approximate $\cvp$}

\noindent
Let $(B,v)$ be an input to $\cvp_{\gamma}$ for some $\gamma$, 
where $B$ is an LLL-reduced upper-triangular matrix.
Let $\langle B_{SNF},T,M \rangle$ denote the tuple returned by the SNF reduction, for parameters $a>1/2,b$.
Suppose that for $x_0\in L(B_{SNF})$, 
$v\in [\det(B)]^n \cap \Z^n$ 
we have
$\|x_0 - Tv\| \leq \gamma \cdot \dist(Tv, L(B_{SNF}))$, for $\gamma = n^{a - 1/2}$.
Then the vector $x_{out} := B(B_{SNF}^{-1} M) x_0\in L(B)$ has
$$
\| x_{out} - v\| \leq \gamma \cdot \dist(v,L(B)).
$$
\end{corollary}

\begin{proof}
Denote $L = L(B), L_{SNF} = L(B_{SNF})$.
By the triangle inequality:
\be
\left\|x_{out} - v\right\| =
\left\| x_{out} - x_0/T +x_0/T - v\right\| 
\leq
\left\| x_0/T - v \right\| + \| x_0/T - x_{out}\|
\ee
Since
\be
x_0\in L_{SNF} \ \mbox{ and }
\| x_0 - T v \| \leq \gamma \cdot \dist(Tv, L_{SNF})
\ee
then together with the above we have that $x_{out}\in L$ and
\begin{align}
\|x_{out} - v\| 
&\leq 
\frac{1}{T} \left\| x_0 - Tv \right\| + \| x_0/T - x_{out}\|\\
&\leq
\frac{1}{T} \gamma \cdot \dist(Tv, L_{SNF})  + \| x_0/T - x_{out}\|\\
&\leq\label{eq:align3}
\gamma \cdot \dist(v,L) + \| x_0/T - x_{out}\|,
\end{align}
where in the last inequality we used again Lemma \ref{lem:snf}.
Since $v$ is an integer vector, we can assume that the length of $x_0$ is bounded by:
\be
\| x_0 \| \leq \gamma \|Tv\| = \gamma T \|v\|,   
\ee
hence
\be\label{eq:bound1}
\| x_0 / T \| \leq \gamma \|v \| \leq \gamma \det(B) \sqrt{n} = n^a \det(B).
\ee
By assumption $x_0\in L_{SNF}$ so $x_0/T \in L(B_{SNF}/T)$.
Hence,
by Lemma \ref{lem:snf} and together with the above equation \ref{eq:bound1} then $x_{out}\in L(B)$ and we have
\be
\left\|x_{out} - x_0/T\right\| =  \OO(n^{-b}).
\ee
Plugging this inequality into Equation \ref{eq:align3} implies the proof.

\end{proof}

\subsection{Uniform distribution on SNF-dual}

\begin{fact}\label{fact:1}
Let $L\subseteq \F_N^n$ be some SNF lattice, $\det(L) = N$, and
let ${\cal D}^*$ denote the distribution on $NL^*$ defined by sampling $x\sim D_{\F_N^n,s,c}$,
and computing the corresponding element $y\in (NL^*)_N = \Phi_3(x)$.
If $s\geq \eta_{\eps}(L)$ then $\Delta({\cal D}^*, U[(NL^*)_N]) \leq \eps$.
\end{fact} 

\begin{proof}
By  Fact \ref{fact:mr} we have $\Delta((D_{\F_N^n,s,c} \mod {\cal P}(L), U({\cal P}(L))) \leq \eps$.
By Proposition \ref{prop:bij} $\Phi_3$ is a bijection between ${\cal P}(L)$ and $(NL^*)_N$, so therefore:
sampling $x\sim \rho_{\F_N^n,s,c}$, and computing $y = y(x)\in (NL^*)_N$ results in a distribution $D^*$
which is $\eps$ close to $U[(NL^*)_N]$.
\end{proof}

\section{Rank-$1$ $\sis$ with Prime Modulus}

In this work we will use a somewhat different variant of the Ajtai ensemble 
that arises naturally from the SNF reduction.
We define formally our variant of SIS as follows:
\begin{definition}\label{def:sis}

\textbf{Short-Integer-Solution $\sis(N,\delta)$ - homogeneous}

\noindent
Let $N$ be some prime number, and $\delta >0$ some constant.
Fix $n$ as the minimal positive integer
for which $N \leq n^{\delta n}$.
Given are $n$ numbers $g_1,\hdots, g_n \in \F_N$.
The Short-Integer-Solution problem is to find $n$ numbers $h_1,\hdots, h_n\in \F_N$, such that
$$
\sum_{i=1}^n h_i g_i = 0 \in \F_N^n,\mbox{ and } \max_{i=1}^n |h_i| \leq  2 n^{\delta}.
$$ 
\end{definition}
\noindent
Alternatively, $\sis(N)$ asks for a short vector in the lattice
$$
L = L(g)^{\perp} := \{ h\in \F_N^{n}, \ \ \langle (g_1,\hdots, g_n),h \rangle = 0\in \F_N \}.
$$

We note \cite{VV16} that one can formulate $\sis(N)$ also as a non-homogeneous congruence of the following
form
\begin{definition}\label{def:affine}

\textbf{Short-Integer-Solution $\sis(N,\delta)$ - non-homogeneous}

\noindent
Let $N$ be some prime number, and $\delta>0$ some constant.
Fix $n$ as the minimal positive integer
for which $N \leq n^{\delta n}$.
Given are $n$ numbers $g_1,\hdots, g_n \in \F_N$.
The Short-Integer-Solution problem is to find $n+1$ numbers $h_0,h_1,\hdots, h_n\in \F_N$, such that
$$
\sum_{i=1}^n h_i g_i = h_0 \in \F_N,\mbox{ and } \max_{i=0}^n |h_i| \leq  2 n^{\delta}.
$$ 
\end{definition}
Similarly to the homogeneous case we can define $\sis(N)$ as asking for short vectors in the following $n+1$-dimensional lattice:
$$
L = L(g)_1^{\perp} := \{ h\in \F^{n+1}, \ \ \langle (1,g_1,\hdots, g_n),h \rangle = 0\in \F_N \}.
$$

To see why the two problems are equivalent, note that we can implement an oracle ${\cal O}_h$ for the homogeneous
$\sis(N,\delta)$ using oracle calls to ${\cal O}_n$ for the non-homogeneous $\sis(N,\delta/2)$ as follows:
Given $g$, call ${\cal O}_n(g) = \{h_i\}_{i=0}^n$, i.e. 
$$
\sum_i h_i g_i = h_0 \modn.
$$
Then call 
$$
{\cal O}_n(h_0^{-1} g) = \{h_i'\}_{i=0}^n.
$$
Then by definition
$$
\sum_{i=1}^n h_i' h_0^{-1} g_i = h_0'
$$
so 
$$
\sum_{i=1}^n h_i' g_i = h_0 h_0' = h_0' \sum_{i=1}^n h_i g_i,
$$
which implies that 
$$
\sum_{i=1}^n (h_i' - h_0' h_i) g_i = 0 \modn,
$$ 
and since $|h_i|, |h_i'|$ are bounded by $n^{\delta/2}$ for all $i$ this implies that
$\{h_i' - h_0' h_i\}_{i=1}^n$
are bounded by $n^{\delta/2} + n^{\delta/2} n^{\delta/2} \leq 2 n^{\delta}$, hence it
 is a valid solution to ${\cal O}_h(g)$.

The significance of the non-homogeneous version is evident in the following 
equivalence of definitions:
\begin{proposition}\label{prop:equiv}
Let $B$ be an SNF matrix, with $\det(B) = N$.
Then $L(B) = L(g)_1^{\perp}$ where $g$ is the $n-1$-dimensional vector $g = (b_2,\hdots, b_n)$.
\end{proposition}

The random - approximate SIS problem ($\rsis(N)$) is then defined by having $g_1,\hdots, g_n$ be chosen independently and uniformly at random.
We also define an approximation variant called $\sis_{\gamma}$ in which the solution
$\{h_i\}$ must satisfy $max_i |h_i| \leq n^{\delta} \gamma$.
As in previous works, one can establish the existence of a solution to the SIS problem for any input,
using the pigeonhole principle:
\begin{fact}\label{fact:sol}
For any $N,\delta$ the (homogeneous / non-homogeneous) $\sis(N,\delta)$ problem has a solution.
\end{fact}
\begin{proof}
Fix some $h_0\in \F_N$.
By definition, there are at least $n^{\delta n}$ vectors of coefficients 
$(h_1,\hdots, h_{n})$
whose magnitude is at most $n^{\delta}$. Since $N\leq n^{\delta n}$ there are at least two $n$-dimensional coefficient vectors $h \neq h'$
for which $\sum_{i=1}^n h_i g_i = \sum_{i=1}^n h_i' g_i = h_0 \modn$.
Hence the vector $(h_1-h_1', \hdots, h_{n} - h_{n}')$ is 
a non-zero $n$-dimensional vector which is a valid solution.
\end{proof}
%\begin{proposition}
%
%\textbf{SNF matrix associated with an instance of $\sis(N)$}
%
%\noindent
%The lattice $\Lambda(g)^{\perp}$ is generated by the
%$(n+1)\times (n+1)$-dimensional SNF matrix whose entries $b_2,\hdots, b_{n+1}$
%are equal to $g_1,\hdots, g_n$, respectively.
%\end{proposition}

Next, we adapt an argument that appeared in \cite{AP11} which discussed
the Ajtai ensemble of lattices, for our ensemble of random lattices - $\rsis(N)$:
\begin{lemma}

\textbf{Random $\sis$ lattices are dense}

\noindent
Let $N$ be some prime number, $\delta\leq 1/2$, and consider the lattice defined by $\sis(N,\delta)$,
for some numbers $g_1,\hdots, g_n\in \F_N$:
$$
L = L(g)^{\perp} := \{ z\in \Z^{n}, \ \ \langle (g_1,\hdots, g_n),z \rangle = 0 \}.
$$
With high probability when choosing the numbers $\{g_i\}_{i=1}^n$ 
independently and uniformly at random from $\F_N$ we have that $\lambda_i(L(g)^{\perp})= \Theta(\sqrt{n})$ for all $i\in [n]$.
\end{lemma}

\begin{proof}
By definition $n$ is the minimal positive integer for which $N\leq n^{\delta n}$.  
There exists such $n$ for which $N \geq n^{\delta (n-1)}$.
\footnote{assuming Bertrand-Chebyshev theorem, and sufficiently large $n$.}
The argument has two parts:
First, is that $\det(L) \leq N$ - this is immediate because there are at most N congruence classes of $\Z^{n} / L$ by definition.
Next, for any integer $n$ and constant $\alpha$ let: 
\be
N_{\alpha,n} = \left| \{ z\in \F_N^n, \ \ \|z\| \leq \alpha\sqrt{n} \} \right|.
\ee
By \cite{MO90} for every $\eps$, there exists an $\alpha$ such that $N_{\alpha,n} \leq \eps^n$.
Fix a vector $v\in N_{\alpha,n}$.  We have that 
$\P_{g\sim U[\F_N^{n}]} ( v\in L(g)^{\perp}) = 1/N$.
Hence by the union bound 
\be
\P
\left( \{ z\in \F_N^n, \ \ \|z\| \leq n^{\delta} \} \cap L(g)^{\perp} \neq \Phi 
\right) 
\leq 
\eps^{n} / N \leq \eps^{n} / n^{\delta(n-1)} <<1.
\ee
Together with the assumption $\delta \leq 1/2$ this implies that w.p. $1-\OO(1)$ we get $\lambda_1(L(g)^{\perp})  \geq \alpha n^{\delta}$.
This implies by definition that  $\lambda_i(L) \geq \alpha n^{\delta}$ for all $i\in [n]$.
Together with the fact that $\det(L) \leq N \leq n^{\delta n}$,
it implies that $\lambda_i(L) = \Theta(n^{\delta})$ for all $i\in [n]$.
\end{proof}

\section{Reduction from worst-case GDD to random SIS}

In this section we show
a reduction from a variant of $\cvp$ called Guaranteed-Distance-Decoding, or $\gdd$,
(see Definition \ref{def:gdd})
to random SIS.
Let ${\cal A}$ denote an algorithm for $\sis(N)$.
We define the following algorithm:
\begin{algorithm*}${\cal B}$
\item
Input: An integer matrix $B$,  and vector $v\in \Z^n$, and parameter $\delta>0$.
\begin{enumerate}
\item
Reduce $B$ via LLL, denote by $B_{LLL}$.  
\item
Decompose $B_{LLL}$ by QR decomposition, $B_{LLL} = Q \cdot R$,
set $\hat{v} = Q^{\dag} v$.
\item 
Reduce $R$ to SNF form: $\langle B_{SNF},M,T \rangle$.
Denote $L_{SNF} = L(B_{SNF})$,  $N = \det(L_{SNF})$, $L_N = L_{SNF} \cap \F_N^n$,  $NL^* = L(N \cdot B_{SNF}^{-T})$.
\item
Put 
$m$ as the minimal positive integer for which $m^{\delta m} \geq N$.
Let $\Phi$ denote the minimal number for which: $\eta_{\eps}(L) \leq \Phi$ for $\eps = m^{-5}$.
Set $s = T \cdot \Phi$.
\item
Choose $c$ uniformly at random from $\F_N \cap [- m^{1 + \delta}, m^{1 + \delta}]$.
Choose $u_{rand}\in L_N$ uniformly at random from $L_N$.
Put $v_{target} = (c^{-1}\cdot T \cdot \hat{v} + u_{rand}) \modn$.
\item
Repeat $m$ times:
\begin{enumerate}
\item
Sample $x_i = (x_{i,1},\hdots, x_{i,n})\sim D_{\F_N^n,s,v_{target}}$.
\item
Compute $y_i = \Phi_3(x_i)$.
Denote by $a_i$ the first coordinate of $y_i$.
\end{enumerate}
\item
Put $\{\alpha_i\}_{i=1}^{m} = {\cal A}(a_1,\hdots, a_m)$.
\item
If ${\cal A}$ fails or $\sum_{i=1}^m \alpha_i \neq c$ return FAIL.
\item
Compute $x_0 = \sum_{i=1}^m \alpha_i (y_i + x_i) - c u_{rand}\modn$.
\item
Return $x_{out} := Q \cdot R \cdot (B_{SNF}^{-1}M^{-1} \cdot x_0) $.
\end{enumerate}
\end{algorithm*}

The theorem below is stated so that for an oracle ${\cal A}$ that returns a correct answer w.p. $1-o(1)$ the algorithm ${\cal B}$
computes the correct answer w.p. $1/\poly(n)$.
However, by making multiple (parallel) calls to the oracle ${\cal A}$, and algorithm ${\cal B}$, the probability of success can be amplified to an arbitrary  constant, while using an oracle that returns a correct answer w.p. $1/\poly(n)$.
We note that by Lemma \ref{lem:snf}, the minimal $m$ for which $m^{\delta m} \geq N$,
where $N = \det(B_{SNF})$ is polynomial in $n$, i.e. $m = \poly(n)$.
Thus, for simplicity, in our main theorem we state the probability of success as a function of $m$
instead of $n$.
\begin{theorem}
Let $(B,v)$ be an input to $\gdd$
where $B$ is an $n\times n$ integer.
Suppose that ${\cal A}$ returns w.p. at least $1-m^{-3}$ a solution to
$\sis(N,\delta)$.  Then for $x_{out} = {\cal B}(B,v)$ we have $x_{out}\in L$ and w.p. $\Omega(m^{-2})$ we have:
$$
\|x_{out} - v\| \leq \Phi \cdot \dist(v,L) \cdot \OO(n^{1.5+\delta} \cdot \max\{n,\log\det(B)\}^{1+\delta}).
$$
\end{theorem}

\begin{proof}
$B_{LLL}$ is LLL-reduced, hence the matrix $R$ as a basis for $L = L(R)$ is trivially also LLL-reduced.
$R$ is thus LLL-reduced and upper-triangular which means we can invoke Lemma \ref{lem:snf} w.r.t. $R$, and derive
the tuple $\langle B_{SNF},M,T\rangle$ which has the properties of the lemma.
In particular, by Corollary \ref{cor:cvp} 
if 
\be
(\star)\ \ x_0\in L_{SNF} \ \mbox{ and }
\| x_0 - T \hat{v} \| \leq \gamma \cdot T\Phi
\ee
then
\be\label{eq:star1}
 R \cdot (B_{SNF}^{-1}M^{-1} \cdot x_0)\in L(R) \mbox{ and }
\| R \cdot (B_{SNF}^{-1}M^{-1} \cdot x_0) - \hat{v}\| 
\leq
\gamma \Phi + \OO(n^{-k}),
\ee
and since $Q$ is a unitary matrix then
\be
x_{out} \in L(B), \ \ 
\| x_{out} - v \|
=
\| R \cdot (B_{SNF}^{-1}M^{-1} \cdot x_0) - \hat{v}\|
\leq
\gamma \Phi + \OO(n^{-k}).
\ee
We will first show that when ${\cal B}$ terminates then the output vector has property $(\star)$.
Then, we will show that ${\cal B}$ succeeds with high probability.

\item
\textbf{Property $(\star)$ - quality of estimate:}

\noindent
By Claim \ref{cl:dual} we have:
\be\label{eq:inlattice}
\forall i\in [m], \ \ x_i + \Phi_3(x_i) = x_i + y_i \in L.
\ee
Since $x_i\sim D_{\F_N^n,s,v_{target}}$, for $v_{target}\in \F_N^n$ we can write 
\be\label{eq:expand}
x_i = c^{-1} T \hat{v} + u_{rand} + {\cal E}_i,\ \ 
{\cal E}_i \sim D_{\F_N^n,s,0}.
\ee
Since we choose $m$ such that $N\leq \sqrt{m}^{m}$ then by Fact \ref{fact:sol} there exists a short solution, i.e. whenever ${\cal A}$
succeeds it returns
a set of coefficients $\alpha_1,\hdots, \alpha_m$, $|\alpha_i| \leq \beta\sqrt{m}$, 
for some constant $\beta$, and in addition 
$$
\sum_{i=1}^m \alpha_i a_i = 0.
$$
Assume from now on that this is the case.
Using again Equation \ref{eq:nl} implies that
\be\label{eq:aout}
\sum_{i=1}^m \alpha_i y_i = 0 \in \F_N, \ \ |\alpha_i| \leq \beta m^{\delta}, \ \forall \ 1\leq i\leq n
\ee
Now, assume that $\sum_{i=1}^m \alpha_i = c$, in which case the algorithm ${\cal B}$ 
declares success and consider the vector: 
\be
x_0 := 
\sum_{i=1}^m \alpha_i (x_i + y_i)- c u_{rand} .
\ee
On one hand: by definition $x_0 \in L_{SNF}$, because the $\alpha_i$'s and $c$ are integers, $u_{rand}\in L_N$, and by Equation \ref{eq:inlattice} $x_i + y_i\in L_N$ for all $i$.
On the other hand by Equation \ref{eq:expand}
\be\label{eq:1}
T \hat{v} - x_0
=
T \hat{v} - 
\sum_{i=1}^m \alpha_i y_i - 
\sum_{i=1}^m \alpha_i (c^{-1} T \hat{v}+u_{rand}) - 
\sum_{i=1}^m \alpha_i {\cal E}_i + c u_{rand}
\ee
In addition, by the assumption that $\sum_{i=1}^m \alpha_i = c$ we get:
\be
\sum_{i=1}^m \alpha_i c^{-1} T \hat{v} = 
T \hat{v} c^{-1} \cdot \sum_{i=1}^m \alpha_i = T \hat{v} \cdot c^{-1} \cdot c = T \hat{v} ,
\ee
and similarly 
\be
\sum_{i=1}^m \alpha_i u_{rand} = c u_{rand}.
\ee
Hence
\be
T \hat{v} - x_0 = \sum_i \alpha_i {\cal E}_i
\ee
By Equation \ref{eq:expand} ${\cal E}_i\sim D_{\F_N^n, s,0}$ for all $i$, so by Fact \ref{fact:ban}
\be
\forall i, \ \ 
\P( \|{\cal E}_i\| \leq s \sqrt{n}) = 1-2^{-\Omega(n)},
\ee
and so by the union bound:
\be
\P( \forall i, \ \|{\cal E}_i\| \leq s \sqrt{n}) = 1- 2^{-\Omega(n)} = 1 - o(1).
\ee
Next, by Equation \ref{eq:aout} the random variables $\alpha_i$ are bounded by $m^{\delta}$ hence
\be
\P( \|T \hat{v} - x_0 \| \leq s \cdot m^{1+\delta} \cdot \sqrt{n}) = 1 - o(1).
\ee 
We note that in the above we assume a worst-case scenario where
the summation $\sum_i \alpha_i {\cal E}_i$ is coherent, i.e. the $\alpha_i$'s
are completely correlated with the ${\cal E}_i$'s.
Together with our choice $s = T \Phi$ this implies that
\be
\P( \| T \hat{v} - x_0 \| \leq T \phi m^{1+\delta} \sqrt{n}) = 1 - o(1).
\ee
%By our choice of parameters $s = \Phi \cdot T$, and since $L$ is an integer matrix this implies that:
%\be
%s \leq \Phi \cdot \dist(Tv,L_{SNF}) \leq T \Phi \cdot \dist(v,L),
%\ee 
%where the last inequality follows again from Lemma \ref{lem:snf}.
Since  $m$ is the minimal integer for which
$m^{\delta m} \geq N$ 
then by Bertrand-Chebyshev theorem for sufficiently large $m$ we have:
\be
m^{\delta (m-1)}\leq N.
\ee
In addition, by Lemma \ref{lem:snf} we have $T\leq \max\{ \det(B), 2^{\OO(n)} \}$ and 
$N = \OO(\det(B) T^n)$.
Hence $m = \OO(n \max\{n,\log\det(B)\})$.
Thus, w.p. $1 - o(1)$ we have 
\be
\| T \hat{v} - x_0\| = \OO( n^{1.5+\delta} T \Phi \max\{n,\log\det(B)\}^{1+\delta}),
\ee
in which case by Equation \ref{eq:star1}
\be\label{eq:out}
\| v - x_{out} \| = \Phi \cdot \OO( n^{1.5+\delta} \max\{n,\log\det(B)\}^{1+\delta}) .
\ee

\item
\textbf{Probability of success:}

\noindent
Now we would like to lower-bound the probability of the algorithm succeeding.
By our choice of $s = T \cdot \eta_{\eps}(L)$, 
and Lemma \ref{lem:snf} we have: $s \geq \eta_{\eps}(L_{SNF})$.
Hence, by Fact \ref{fact:1} we have that 
\be
\Delta(y_i , U[(NL^*)_N]) \leq \eps.
\ee
By equation \ref{eq:nl} this implies that $\Delta(a_i, U[N]) \leq \eps$,
and since the $a_i$'s are independent then
\be
\Delta((a_1,\hdots,a_m), U[N^m]) \leq 1 - (1-\eps)^m \leq m^{-3},
\ee
where the last inequality follows from our choice of $\eps = m^{-5}$.
By assumption, algorithm ${\cal A}$ for $\sis(N)$ returns a correct answer w.p at least $1- m^{-3}$, for 
a uniformly random vector $(a_1,\hdots, a_m)$.
Therefore
\be
\P({\cal A} \mbox{ succeeds } ) \geq 1 - m^{-3} -  m^{-3} = 1-\OO(m^{-3}). 
\ee
Next, we would like to make sure that ${\cal A}$ is not "adversarial", and given $c$, returns
$\alpha_1,\hdots, \alpha_m$ whose sum is different than $c$ almost all the time.
Using Fact \ref{fact:2}, for any $c_1,c_2\in [N]$ we have:
\be\label{eq:ratio}
\Delta \left( \P(\alpha_1,\hdots,\alpha_m | c_1) ,  \P(\alpha_1,\hdots,\alpha_m | c_2) \right) \leq m^{-3}.
\ee
Hence, given a choice of $c$ by the algorithm,
the probability that the same $c$ is going to be the value of $\sum_i \alpha_i$ is lower-bounded by:
\be
\P\left( \sum_{i=1}^m \alpha_i = c | c\right) \geq \max_{c'} \P\left(\sum_i \alpha_i = c'\right) - m^{-3} = \Omega\left(\frac{1}{m^{1 + \delta}}\right),
\ee
where the inequality follows because
$|\alpha_i| \leq 2m^{\delta}$ for each $i$, and so $\left|\sum_i \alpha_i \right| \leq 2m^{1 + \delta} $.
By the union bound we thus have:
\be
\P({\cal B} \mbox{ succeeds }) =
\P \left({\cal A} \mbox{ succeeds } \wedge \sum_{i=1}^m  \alpha_i= c\right)
=
\Omega(m^{-1-\delta}) - O(m^{-3})= \Omega(m^{-2}).
\ee
\end{proof}

\begin{fact}\label{fact:2}
For any $c_1,c_2\in [N]$:
$$
\Delta
\left(
\P(\alpha_1,\hdots,\alpha_m | c_1), 
\P(\alpha_1,\hdots,\alpha_m | c_2)
\right) 
\leq m^{-3}.
$$
\end{fact}
\begin{proof}
Since $y_i$ is uniquely determined from $x_i$ for each $i$, and $(\alpha_1,\hdots,\alpha_m)$ is determined (possibly probabilistically) from the pairs 
$(y_i,x_i)$ for all $i\in [n]$ then it is sufficient to show that for any $c_1\neq c_2$, we have
\be
\Delta
\left(
\P(x_1,\hdots,x_m | c_1), 
\P(x_1,\hdots,x_m | c_2)
\right) 
\leq m^{-3}.
\ee
By definition, we have:
\be
\P(x_1,\hdots, x_m | c) \propto \prod_{i=1}^m e^{-\pi \|x_i - v c - u_{rand}\| / s^2}.
\ee
Since $u_{rand}$ is uniform on $L_N$, 
and chosen independently of $c$, then it is sufficient to consider the a-posteriori probabilities
of the vectors $x_i$ modulo ${\cal P}(L)$ for each $i$:
\be
\P(x_1^L ,\hdots,x_m^L  | c), \ \ x_i^L := x_i\  \mod {\cal P}(L).
\ee
Finally, since the $x_i$'s are independent given $c$ then:
\be
\P(x_1^L,\hdots,x_m^L | c)
=
\prod_{i=1}^m \P( x_i^L |c)
\ee
Hence, for $c_1\neq c_2$ we have:
\be\label{eq:allx}
\Delta
\left(
\P(x_1,\hdots,x_m | c_1), 
\P(x_1,\hdots,x_m | c_2)
\right) 
=
\Delta \left(
\prod_{i=1}^m \P( x_i^L |c_1) ,
\prod_{i=1}^m \P( x_i^L |c_2)
\right).
\ee
Since $s \geq \eta_{\eps}(L_{SNF})$ then by Fact \ref{fact:mr} and the triangle inequality we have:
\be
\forall i\in [m] \ \ 
\Delta \left( \P( x_i^L |c_1), \P( x_i^L |c_2) \right) \leq 2\eps
\ee
which implies by Equation \ref{eq:allx}
\be
\Delta
\left(
\P(x_1,\hdots,x_m | c_1), 
\P(x_1,\hdots,x_m | c_2)
\right) 
\leq
1 - (1-2\eps)^m
= O(m^{-4}),
\ee
where the last equality follows from our choice $\eps = m^{-5}$.
\end{proof}

\noindent
We note that by a simple modification of the algorithm above one can obtain a reduction from approximate $\sivp$ to $\rsis(N)$:
\begin{theorem}
Given is an integer matrix $B$.
Let $\Phi$ denote the number defined by: $\eta_{\eps}(L) \leq \Phi$ for $\eps = m^{-5}$.
Suppose that ${\cal A}$ returns w.p. at least $1-m^{-3}$ a solution to
$\rsis(N,\delta)$.  
There exists an algorithm ${\cal B}$ such that ${\cal B}^{\cal A}$ returns w.p. $1 - o(1)$ 
a vector $x_{out}\in L$, such that
$$
\|x_{out} \| \leq \Phi \cdot \OO(n^{1.5+\delta} \max\{n,\log\det(B)\}^{1+\delta}).
$$
\end{theorem}
\begin{proof}
We slightly modify algorithm ${\cal B}$ as follows: we set $v_{target} = u_{rand} = v= 0$ in step (3).
In Step (6) we terminate successfully whenever ${\cal A}$ succeeds (i.e. regardless of the sum of coefficients).
Similar to Equation \ref{eq:out} we have:
\be
x_{out}\in L, 
\P( \|x_{out}\| \leq \Phi \cdot \OO( n^{1.5+\delta} \max\{n,\log\det(B)\}^{1+\delta})) = 1 -o(1).
\ee
Hence, sampling $n$ such vectors $x_{out}$ independently returns w.h.p. a set of $n$ linearly-independent vectors,
with the desired approximation ratio.
\end{proof}

\section{Acknowledgements}
The authors are thankful to 
Oded Regev and Vinod Vaikuntanathan for their important insights regarding $\sis$ and to
the following people for useful discussion and insightful comments:  
Shafi Goldwasser, Sean Hallgren, David Jerison,
Daniele Micciancio, Noah Stephens-Davidowitz.
LE is thankful to the Templeton foundation for their support of this work.
PWS was supported by the US Army Research Laboratory's Army Research Office through grant number W911NF-12-1-0486, the NSF through grant number CCF-121-8176, and by the NSF through the STC for Science of
Information under grant number CCF0-939370.

\end{document}